\documentclass[11pt]{article}

\date{}

\usepackage{url}
\usepackage{cite}
\usepackage{plain}

\usepackage{wrapfig}
\usepackage{graphics}
\usepackage{picins,graphicx}
\usepackage[english]{babel}
\usepackage[leqno]{amsmath}
\usepackage{amssymb}
\usepackage{verbatim}
\usepackage[mathscr]{eucal}

\usepackage{amstext}
\usepackage{makeidx}
\usepackage{amsthm}

\usepackage{hyperref}
\hypersetup{colorlinks,%
citecolor=red,%
linkcolor=blue,%
urlcolor=black}

\makeindex

\newtheorem{theorem}{Teorema} 

\newtheorem{proposition}{Proprieta'}
\newtheorem{definition}{Definizione}
\newtheorem{notation}{Nota}
\newtheorem{ex}{Esercizio} 
\newtheorem{esempio}{Esempio}

\newcommand{\vs}{\vspace{3mm}}
 
\newcommand{\no}{\noindent} 

\newcommand{\beq}{\begin{equation}} 
\newcommand{\eeq}{\end{equation}}

\newcommand{\bex}{\begin{ex}} 
\newcommand{\eex}{\end{ex}} 

\newcommand{\bese}{\begin{esempio}} 
\newcommand{\eese}{\end{esempio}} 

\newcommand{\bpro}{\begin{proposition}} 
\newcommand{\epro}{\end{proposition}}

\newcommand{\bthe}{\begin{theorem}} 
\newcommand{\ethe}{\end{theorem}}

\newcommand{\bnote}{\begin{notation}} 
\newcommand{\enote}{\end{notation}}

\newcommand{\bdefi}{\begin{definition}} 
\newcommand{\edefi}{\end{definition}} 

\newcommand{\bc}{\begin{center}} 
\newcommand{\ec}{\end{center}}

\newcommand{\mail}[1]{\href{unina:#1}{\texttt{#1}}}

\usepackage{pdfsync}

\author{Monica De Angelis, P. Renno \thanks{Univ. of Naples  "Federico II", Faculty of Engineering, Dip. Mat. Appl. "R.Caccioppoli", \newline
 Via Claudio n.21, 80125, Naples, Italy.
\newline\mail{modeange@unina.it}}}

\title{ Asymptotic effects of boundary perturbations in excitable systems}

\begin{document}
\maketitle

\begin{abstract}

A Neumann  problem in the strip for the   Fitzhugh Nagumo system  is considered. The  transformation in a non linear integral equation permits to deduce a priori estimates for the solution. A complete asymptotic analysis shows that for large $ t $    the  effects of the initial data vanish while the effects of boundary disturbances  $ \varphi_1 (t), $ $ \varphi_2(t) $ depend on the properties of the data. When $ \varphi_1 , $ $ \varphi_2 $  are convergent for large $ t $ , the solution is everywhere bounded; when  $ \dot \varphi_i \,\,   \in L^1 (0,\infty) (i=1,2)$  too, the effects are vanishing.
\end{abstract}

\section{Introduction}

Aim of the paper is  the asymptotic analysis of the solution of the Fitzhugh Nagumo  system  (FHN) for a strip ploblem with  Neumann conditions. Some applications are related to the theory of excitable systems; in particular the cases of pacemakers \cite{ks} and when two species reaction-diffusion systems is governed by flux boundary condition \cite{m2}. Moreover, Neumann conditions are applied also  in the distibuted FHN system.\cite{ns}.Several aspects concerning the FHN model are discussed in previous paper \cite{mda2010,dr08,drw}. Moreover, owing to the equivalence between the FHN model and the equation of superconductivity, other applications have been analyzed. 
\cite{ddf} - \cite{df13}, \cite{acscott,acscott02}.

The present paper analyzes a transformation of the FHN  model in a suitable non linear integral-equation (see \ref{36}) whose kernel is a Green function which has numerous  basic properties typical of the diffusion equation. Those properties imply a priori estimates  and so theorems on behaviour of the solution for large t can be obtained.   

\section{Statement of the  problem}

Let $ \, u(x,t)\, $  be a trasmembrane potential and let $\,v(x,t)\,$  be a  variable associated with the contributions to the membrane current from sodium , potassium and other ions. The  well known FHN system   \cite{i,ks,m1,m2,acscott,acscott02} is 

\vs 
\beq     \label{21}
  \left \{
   \begin{array}{lll}
    \displaystyle{\frac{\partial \,u }{\partial \,t }} =\,  \varepsilon \,\frac{\partial^2 \,u }{\partial \,x^2 }
     \,-\, v\,\,  + f(u ) \,  \\
\\
\displaystyle{\frac{\partial \,v }{\partial \,t } }\, = \, b\, u\,
- \beta\, v\,
   \end{array}
  \right.
 \eeq

\vs \vs \noindent where $ \varepsilon \, > \, 0\, $ is a diffusion coefficient related to the  axial current in the axon, while $ \, b\,$ and $\, \beta \, $ are positive constants that characterize the model's kinetic. Further

\beq      \label{22}
f(u)\, =\,  u \, ( a- u ) \, ( u-1)    \qquad  \,(0<a<1).\, 
\eeq

\vs
Assuming  $ T  $ as an arbitrary positive value, a typical example of problems which takes into account either initial perturbations and  boundary perturbations is defined in 

\[
\,   \Omega \, \equiv \{\,(x,t) : \, 0\,\leq \,x \,\leq L \,\,;  \ 0<t\,<T\, \} \]

 \no by

\vs \no 
\beq      \label{23}
u (x,0)\, = u_0(x),\,\qquad \,\,\, \,v (x,0)\, = v_0(x)\, \,\,\, 
\eeq

\vs \no 
with the Neumann conditions

\beq   \label{24}
u_x(0,t)\,=\, \varphi  _1(t)  \qquad u_x(L,t)\,=\,\varphi _2(t). 
\eeq

\vspace{3mm}It can be easiy verified (see,f.i\cite{mda2010,dr08}) that the problem can be analyzed by means of  an  integral differential problem  with a single  unknown function $ u(x,t) $. In fact, if $ F $  denotes the function:

\beq   \label{25}
\,\, F(x,t,u)\, =\,u^2\, (\,a+1\,-u\,)   \, - \, v_0 \,\, e^{-\,\beta \,t}
\eeq

\vs \no  by (\ref{21}) - (\ref{22})   one has

\beq   \label{26}
   u_t -  \varepsilon  u_{xx} + au +b \int^t_0  e^{- \beta (t-\tau)}\, u(x,\tau) \, d\tau \,=\, F(x,t,u) \, \qquad  (x,t) \in \Omega \,  \\\eeq    

\vs \no  with $u $ that has to satisfy the initial - boundary conditions  $(\ref{23})_1,$ (\ref{24}).

\vs
\vs\no As soon as $u(x,t) $ is determined, the $v(x,t)\,$ component will be given by

\beq      \label{27}
v\,(x,t) \, =\,v_0 \, e^{\,-\,\beta\,t\,} \,+\, b\, \int_0^t\, e^{\,-\,\beta\,(\,t-\tau\,)}\,u(x,\tau) \, d\tau,
\eeq

\no where $ v_0 $ is  defined in $(\ref{23})_2$.

\vs \no When source term $ F $ in  (\ref{26}) is  a prefixed  function depending only on $\, x\,$  and $\, t,\, $ then  the initial-boundary problem  (\ref{26}),$\,\,\,(\ref{23})_1, \,\,(\ref{24})$ is linear and it can be solved explicitly by means of the Laplace transform. Moreover, when $ F $  depends on the unknown $ u(x,t)  $ too, then by (\ref{26}) one obtains  an integral equation useful to study the differential problem.

\section{Previous results}

The fundamental solution $ \, K(x,t) \, $ of the parabolic operator defined by (\ref{26}) has  been already determined explictly in \cite {dr08} and is given by

\beq     \label {31}
K(r,t)=  \frac{1}{2 \sqrt{\pi  \varepsilon } }\,\biggl[ \frac{ e^{- \frac{r^2 }{4 t}-a\,t\,}}{\sqrt t}
 -\,b \int^t_0  \frac{e^{- \frac{r^2}{4 y}\,- a\,y}}{\sqrt{t-y}}  \,\ e^{-\beta ( t \,-y\,)}  J_1 (2 \sqrt{\,by(t-y)\,}\,)dy \biggr]
\eeq

\vs\no where $\, r\,= |x| \, / \sqrt \varepsilon \, \, $ and   $ J_n (z) \,$    denotes the Bessel function of first kind and order $\, n.\,$ Moreover, one has  \cite{dr08}:

\vspace{3mm}\begin{theorem}

For all $t>0$, the Laplace transform of  $\,K (r,t)\, \,$  with respect to $\,t\,$ converges absolutely in the half-plane $ \Re e  \,s > \,max(\,-\,a ,\,-\beta\,)\,$ and it results:

\vs
\beq      \label{32}
\,\hat K\,(r,s)  =\,\int_ 0^\infty e^{-st} \,\, K\,(r,t) \,\,dt \,\,=  \, \frac{e^{- \,r\,\sigma}}{2 \, \sqrt\varepsilon \,\sigma \,  }
\eeq

\vs \no with $\,\sigma^2 \ \,=\, s\, +\, a \, + \, \frac{b}{s+\beta}.$ 
\end{theorem}

Let us  now consider the following  Laplace transforms with respect to $\,t\,$:

\vs
\[
\hat u (x,s) \, = \int_ 0^\infty \, e^{-st} \, u(x,t) \,dt \,\,, \,\,\,  \,\hat F (x,s)   \, = \int_ 0^\infty \,\, e^{-st} \,\, F\,[x,t,u (x,t)\,] \,dt \,,\
\]

\vs \no  and let $\hat \varphi_1(s ), \,\,\, \hat \varphi_2(s)\,\, $  be the  ${ L} $  transforms of the  data $ \varphi_i(t ) \,\,(i=1,2).\, $

\vs \noindent  Then the Laplace transform of the problem (\ref{26}),$\,\,\,(\ref{23})_1, \,\,(\ref{24})$  is formally given by:

\beq   \label{33}
\left \{
   \begin{array}{lll}
  \hat u_{xx}  \,\,- \dfrac{\sigma^2}{\varepsilon} \,\,\hat u =\, -\,\dfrac{1}{\varepsilon} \,\,[\, \,\hat F(x,s,\hat u(x,s)) +u_0(x)\,\,]\\    
\\
  \,\hat u_x(0,s)\,=\, \hat \varphi_1\,(s)\qquad \hat u_x(L,s)\,=\,\hat \varphi_2\,(s).
   \end{array}
  \right.
\eeq

\vspace{3mm}\noindent   If one introduces the following theta function

\beq\,  \label{34}
\begin{split}
\displaystyle
\hat \theta \,(\,y,\sigma)\,= 
\\ & \frac{1}{2 \,\, \sqrt\varepsilon \,\,\,\sigma  } \, \biggl\{\, e^{- \frac{y}{\sqrt \varepsilon} \,\,\sigma}+\, \sum_{n=1}^\infty \,\, \biggl[ \,e^{- \frac{2nL+y}{\sqrt \varepsilon} \,\,\sigma} \, +\, e^{- \frac{2nL-y}{\sqrt \varepsilon} \,\,\sigma}\,
\biggr] \, \biggr\}
\\ \\& =\dfrac{\cosh\,[\, \sigma/\sqrt{\varepsilon} \,\,(L-y)\,]}{\,2\, \, \sqrt{\varepsilon} \,\, \sigma\,\,\, \sinh\, (\,\sigma/\sqrt{\varepsilon}\,\, \,L\,)}\,\,=
\\&
\end{split}\eeq

\vspace{3mm}\noindent then, by  (\ref{33}) and (\ref{34}) one deduces:

\beq     \label{255}
\begin{split}
\hat u (x,s) = &\,\int _0^L \, [\,\hat \theta\,(\,|x-\xi|, \,s\,)\,+\,\,\,\hat \theta\,(\,|x+\xi|,\, s\,)\,] \, \,[\,u_0(\,\xi\,) \,+\,\hat F(\,\xi,s, \hat u(x,s)\,]\,d\xi\,  
\\
\\ & - \,  \,\,\ 2 \,\,\varepsilon \, \,\hat \varphi_1 \,(s) \,\, \hat  \theta (x,s)\,+ \, 2 \,\, \varepsilon  \,\, \hat \varphi_2 \, (s)\,\,\hat  \theta \,(L-x,s\,).\, \,\, 
\end{split}
\eeq

\vspace{3mm}\noindent  Owing to dependence of source term $ F $ on the unknown, obviously all this is purely formal. However, if one puts

\beq      \label{35}
\left \{
   \begin{array}{lll}
 \theta (x,t) \,=\displaystyle  \sum_{n= -\infty }^\infty \, \, K(x \,+2nL,\,t) \, 
 \\
 \\
 G(x,\xi, t) \, = \,  \theta \,(\,|x-\xi|,\, t\,)\,+ \,  \theta \,(\,|x+\xi|,\,t\,)
     \end{array}
  \right.   
\eeq

\vspace {3mm}\noindent by (\ref{255})  one deduces \cite{mda2010}:

\beq   \label{36}
\begin{split}
 u(\, x,\,t\,)\, =
 \\ & \,\,\int^L_0 \, G(x,\xi,t)\, \,u_0(\xi)\,d\,\xi\,\,- \,2 \, \varepsilon \,\int^t_0 \theta\, (x,\, t-\tau) \,\,\, \varphi_1 (\tau )\,\,d\tau\,
\\ \\& +\, 2\,\, \varepsilon \int^t_0 \theta\, (L-x,\, t-\tau) \,\,\, \varphi _2 (\tau )\,\,d\tau\,\\ \\
 &+\, \,\int^t_0 d\tau\int^L_0 \, \,G(x,\xi,t-\tau)\, \,\,\, F\,[\,\xi,\tau, u(\,\xi,\tau\,)\,] \,\,d\xi.
 \end{split}
\eeq

\vspace{3mm}\noindent which represents an integral equation for the unknown $ u(x,t). $

\section{Basic estimates for the kernels $ K(x,t)  $  and $ \theta(x,t) $}

The behaviour for large $ t $  of the terms depending on the initial data and the source $ F $  has been already analyzed in \cite{mda2010}\cite{mda13}. Now the effects of the boundary perturbations $ \varphi_1,\,\, \varphi_2 $  will be estimated. For this an appropriate analysis of the kernels   $\  K(x,t) \, $   and $ \theta(x,t)  $ will be considered.

As for  $\  K(x,t), \, $   in \cite{dr08} has been proved that

\vspace{3mm}

\beq               \label{42}
|K| \, \leq \, \frac{e^{- \frac{r^2}{4 t}\,}}{2\,\sqrt{\pi \varepsilon t}} \,\, [ \, e^{\,-\,at}\, +\, b t \,E(t)\, ] 
\eeq

\noindent where

\vs
\beq      \label{41} 
E(t) \,=\, \frac{e^{\,-\,\beta t}\,-\,e^{\,-\,at}}{a\,-\,\beta}\,\,>0\,.
\eeq

Further, it results too:


\vs
\beq               \label{45}
\int_\Re\,\,|K(x-\xi,t)|\,d\xi\,\,\leq \, e^{\,-\,at}\, +\, \sqrt b\, \pi \,t \,  \, e^{\,-\,\omega \, t } \,
\eeq

\vs
\beq               \label{46}
   \int_0^t\,d \tau\, \int_\Re |K(x-\xi,t)| \, d\xi \leq \,  \beta_0. 
\eeq

\vspace{3mm}\noindent with

\beq \label{44}
 \, \omega = min \,  ( a, \beta )  \qquad \qquad  \beta_0 \,= \frac{1}{a}\, +\, \pi \sqrt b \, \, \displaystyle {\frac{a+\beta}{2(a\beta)^{3/2}}}.\,\,
\eeq

\vs \no Now, if  $ \Gamma (x)   $ is the gamma function and $ \zeta(x)  $  the Riemann's Zeta function,   let

\beq   \label{411}
 C_0\, = \, \frac{1}{ 2 \sqrt{\varepsilon \, \omega }}\,+\,\frac{ b\,\,\Gamma ( 3/2) \,\,\omega^{-\, 3/2}}{2 \sqrt{\pi \, \varepsilon } \,\,|a-\beta|} \,\biggl[\,1\,\, +\, \dfrac{C}{b}\,|a-\beta|\,+ \frac{3 \,C}{2\, \omega }\,\,\biggr] 
\eeq

\vs \vs \no with     $C= 2 \varepsilon \,\,\zeta (2) / (\,e  L ^2\,). $

\vs \no  Then, one has the following theorem:

\begin{theorem}
The $\theta (x,t)$ function  defined in $( \ref{35})_1 $ satisfies the following inequalities:

\vs
\beq               \label{47}
\int_0^L |\theta (|x-\xi|,\,t)|\ \, d\xi \leq \,  ( 1\, +\, \sqrt b \,\pi \,t \, ) \,\,e^{- \omega \, t\,}
\eeq

\vs
\beq               \label{48}
   \int_0^t\,d \tau\, \int_0^L |\theta (|x-\xi|,\,t)|\ \, d\xi \leq \,  \beta_0\, \qquad \,  . 
\eeq

\vs \no Furthermore, it results: 

\beq              \label{49}
\lim _{t \to \infty}  \theta ( x, t ) \,\, = \,\,0; \quad \quad \int_0^ \infty \, | \theta ( x,\tau )| \,\, d \tau  \,\, \leq \,\,C_0,
 \eeq

 \vs  \no and

\vs 
\beq              \label{412}
\lim _{t \to \infty} \int _0^t \theta ( x, \tau ) \,\, d\tau \,\,= \frac{1}{2 \, \varepsilon \,\,\sigma _0\,  }\,\,\ \frac{\cosh \sigma_0 \,\,(x-L)}{\sinh\,(  \sigma_0 \, L) }
\eeq

\vs \vs \no  where $ \sigma_0 = \sqrt{\biggl(\,a\,\,+ \dfrac{b}{\beta}\biggr)\dfrac{1}{\varepsilon}}.$

\end{theorem}

\begin{proof}: 
We observ that  properties of $ K(x,t)$ imply that:

\vs 
\beq               \label{413}
\begin{split}
& \int_0^L\,|\theta (|x-\xi|,\,t)|\,\,d\xi\,\, \leq \,\sum_{n= -\infty }^\infty \, \, \int_0^L\,| K(|x-\xi +2nL| \,t)| \,d\xi\,\ \\\\& =\sum_{n= -\infty }^\infty \, \, \int_{x+(2n-1)L}^{x+2nL}\,| K(y,\,t)| \,dy\,\,\,\,\leq  \,\, \,\int_\Re\,\,|K(y,t)|\,d y \,
\end{split}
\eeq

\vs \no and so (\ref{47}) and  (\ref{48}) follow by (\ref{45}) and (\ref{46}).

\vs \no Moreover , it results  

\beq        \label{}
 \sum_{n=- \infty }^\infty \, e^{- \frac{(x+2nL)^2}{4\varepsilon t } }\,\,\leq\,\,1\,+\,  \dfrac {\,2\, t \,\varepsilon }{ e\, L^2} \,\,\zeta (2)  
\eeq

\vs \no and (\ref{42})implies:

\beq
| \theta (x,t) | \,= \dfrac {1\,+\,C \, t \, }{ 2\sqrt{\pi \, \varepsilon \, t }} \, [ \, e^{\,-\,a\, t}\, +\, b t \,E(t)\, ],  \,  
\eeq

\vs \no  consequently one obtains  $( \ref{49})_1$  while considered that

\vspace{3mm}
\beq
 \int^\infty_0 t^\mu  e^{- \omega t} dt = \dfrac{\Gamma (\mu +1)}{\omega ^{\mu +1}}\,\,\, 
\begin{array}
r re(\mu)>-1 \\
re(\omega)>0\,
 \end{array}\,\,;
\,\, \, \int^\infty_0   \dfrac{e^{-at}}{\sqrt{t}}  dt \, = \sqrt{\dfrac{\pi}{a}}\quad a>0,
 \eeq

\vs \vs\no  by means of  (\ref{42}),  $ (\ref{49})_2$ can be deduced. 

\vs \no Further as: 
\vs
\beq                \label{418}
\lim _{t \to \infty} \int _0^t \theta ( x, \tau ) \,\, d\tau \,\,= \lim_{s \to 0}\, \hat \theta( x,s)\qquad \quad\Re e  \,s > \,max(\,-\,a ,\,-\beta\,)\,, 
\eeq

\vs \vs  \no 
by (\ref{34}), one achieves(\ref{412}).

\end{proof}

\section{Asymptotic effects of the boundary data}

 In the following  we will have to refer to a known  theorem on asymptotic behaviour of convolutions. (\cite {b},p 66).

\vs{\em Let  $ h(t) $ and $ g(t) $  be two continuous functions on $ [0,\infty [.$   If they    satisfy the following  hypotheses 

\vs \no 
\beq  \label{hp}
\exists \,\, \displaystyle{\lim_{t \to \infty}}h(t) \, = \, h(\infty)\qquad\exists \,\, \displaystyle{\lim_{t \to \infty}}g(t) \, = \, g(\infty),
\eeq  
 
\vs \no 
\beq \label{hp2}
 \dot  g(t)\,  \in \, L_1  [ \,0, \infty),\eeq

\vs \no then, it results: 

\vs \no 
 \beq      \label{61}
\lim_{t \to \infty} \,\, \int_o^t \, h(t-\tau ) \, \dot g ( \tau ) \, d \tau \,\, = \, \,h(\infty) \,\, [\,\,g(\infty) - g(0)\,\,].
\eeq 
} 

\vs \no According to this, it is  possible to state:

\vspace{3mm}
\begin{theorem} \label{theorem asintotico}
Let  $ \varphi_ i  \,\,\ (i=1,2) \,\,$  be two  continuous functions  which converge for $ t \rightarrow \, \infty . $  In this case one has:  

\vs 
\beq    \label{62}
\lim_{t \to \infty } \,\int_0^t \,\theta \,(x,\tau)\,\,\, \varphi_i \,(t-\tau)\, \,d\,\tau \, = \, \varphi_{i, \infty} \,\,\,\,\,  \frac{1}{2 \, \varepsilon \,\,\sigma _0\,  }\,\,\ \frac{\cosh \sigma_0  \,\,(x-L)}{\sinh\  \sigma_0  \, L }
\eeq

\vs \vs \no where  $ \sigma_0 = \sqrt{\biggl(\,a\,\,+ \dfrac{b}{\beta}\biggr) \dfrac{1}{\varepsilon}}.$
\end{theorem}

\vs \begin{proof}
Let apply (\ref{61}) with  $ g= \int _0^t \theta(x,\tau) d\tau \, \,\,\mbox{and}\,\, f = \varphi_i  \,\,(i=1,2) $. Then, (\ref{62}) follows by $( \ref{49})_2 $ and (\ref{412}).

\end{proof}

\vspace{3mm}

\begin{theorem}\label{th62}
When the data  $\varphi _ i \,\,\, ( i=1,2) $ verify conditions  (\ref{hp}) (\ref{hp2}),  it results: 

\beq    \label{63}
\lim_{t \to \infty } \, [\,\theta (x,t)\,\,\,\, \ast \,\, \varphi_i (t) \, ]\, = \,0\,  \qquad ( i= 1, 2) 
\eeq

\end{theorem}
\vs \begin{proof}

It sufficies to put  $ h= \theta(x,t)  $ and $ g= \varphi_i $ and to apply $(\ref{49})_1$.

\end{proof}

\section{Asymptotic  behaviour of the FHN solution}

Let us denote with $\, f_1 \, \ast
f_2 \, $ the convolution

\[ f_1 ( \cdot, t) \ast \, f_2 ( \cdot ,t) \, = \int_0^t \, f_1 ( \cdot, t)  \, f_2 \,\,( \cdot , t -\tau)  \,d\,\tau \]

\vs \no and  let $ N(x,t) $ be the following known function depending on the data 
$( u_0, v_0, \varphi_1,  \varphi_2)$

\beq    \label{51}
\begin{split}
N(x,t)\, =\,
\\&  -2 \,\varepsilon \, \varphi_1 (t) \, \ast \, \theta (x,t) \,
 + \, 2\, \varepsilon \, \varphi_2 (t) \, \ast \, \theta ( L - x , t) \, 
\\ \\&+\,\int^L_0 \,  \,u_0\,(\xi)\,\, G( x, \xi,t) \, d\xi  \, - \, e^{\,-\, \beta\, t\, } \, \ast  \, \int^L_0 \,  v_0( \xi)\,\, G( x, \xi,t) \, d\xi \,.
\end{split} 
\eeq

\vs \no Owing to  (\ref{25}),  (\ref{27}) and (\ref{36}),the solution related to the initial boundary FHN system \ref{21}-\ref{24} is given by \cite{mda2010}:

 \vs \vs 
\beq     \label{71}
 u (x,t) = \,\int _0^L \, G \, (\, x, \xi  , t-\tau) \,\, \ast \,\{ \,u^2\,(\,\xi ,\,\,\tau\,) [\,a+1\,-u\,(\,\xi ,\,\,\tau\, ) \,]\}\,\,d\xi\, \,+\, N
\eeq

\vspace{6mm}
\beq      \label{72}
\begin{split}
v\,(x,t) \, = &\,\, v_0 \, e^{\,-\,\beta\,t\,} \,+\, b\,  e^{\,-\, \beta\, t\, } \, \ast  \, \,N(x,t)
\\ \\&  +\,b\,\,\, e^{\,-\, \beta\, t\, } \, \ast  \,\int _0^L \, G \, (\, x, \xi  , t-\tau) \,\, \ast \,\{ \,u^2\,(\,\xi ,\,\,\tau\,) [\,a+1\,-u\,(\,\xi ,\,\,\tau\, ) \,]\}\,\,d\xi\, 
\end{split}
\eeq

 \vs \no These formulae represent two  integral equations for $\, u\,$ and $ v\,. $ By means of the estimates deduced in sec.4 it is possible to apply the fixed point theorem in order to obtain existence and uniqueness results\cite{c,mda2010,dmm}. When the Nagumo polinomial (\ref{25}) is approximated by means of its linear part, then (\ref{71}) (\ref{72})  give the explicit solution of the problem.

\vs As for the analysis and the stability of solutions of nonlinear binary reaction - diffusion
systems of PDE's, as well as the existence of global compact attractors, there exists   a large bibliography  .  (see e. g. \cite{drw,i,l,lgns,r}. Moreover, as it is well known,the (FHN) system admits arbitrary large invariant rectangles
$ \Sigma $ containing $(0,0)$ so that the solution $(u,v)$, for all times $t > 0$, lies in the interior of $ \Sigma $
 when the initial data $(u_o,v_o)$ belong to $ \Sigma $.\cite{s}

\no So, letting

\[
\|\,F\,\| \,
= 
\sup _{ \Omega_T\,}\, | \,\,u^2  \,(a+1-u)\,\,|, \, 
\]
\vs
\[
\|\,u_0\,\| \,
= \displaystyle \sup _{ \Omega_T\,}\, | \,\,u_0 \,(\,x,) \,| ;\quad\quad \|\,v_0\,\| \,= \displaystyle \sup _{ \Omega_T\,}\, | \,v \,(\,x) \,\,| \, \]

\vs \no one has:

\begin{theorem}  \label{the71}
For regular solution $ (u,v) $ of the (FHN) model, when  the boundary conditions are homogeneous, ( $\varphi_1\,\,= \varphi_2 \,= \,0\, \, $), the following estimates hold:

\vs 
\beq            \label{74}
\left\{ 
 \begin{array}{lll}                                                   
 \left| u \, \right| \, \leq  2\,[\,\left\| u_0 \right\| \, (1+\pi \sqrt b \, t ) \, e^ {\,-\omega\,t\,}\,+\,\left\| v_0 \right\|\,E(t) \, +\, \beta_0 \,\left\| F \right\|\,] 
   \\
\\
\left| v \, \right| \, \leq  \left\| v_0 \right\|\, e^ {\,-\,\beta\,t\,}\,+\,2\,\bigl[\,b\,(\,\left\| u_0 \right\|\,+\, t\, \left\| v_0 \right\|\,) \, E(t) \, + \, \dfrac{b}{a \beta}\, \left\| F \right\| \,\bigr]
\\ 
   \end{array}
  \right.
 \eeq

\end{theorem}

\vs \no  For  boundary data  different from zero,the asymptotic behaviour of  the solution $ (u,v) $ of FHN system is   established by  theorems \ref{theorem asintotico} and \ref{th62}.

 \vs \vs {\bf In conclution}. When   $ t $ tends to infinity,  the effect due to the initial disturbances $\, (\,u_0, v_0\,) \, $  vanishes  while   the effect of the non linear source is bounded for all $ t. $ Moreover, also the effects determined by boundary disturbance  $ \varphi_1, \varphi_2 $ are vanishing in the hypotheses $ (b) $. Otherwise, they are always bounded.

\vspace{3mm}\section*{Acknowledgments} This  work has been performed under the auspices of Programma F.A.R.O. (Finanziamenti per l' Avvio di  Ricerche
Originali, III tornata) ``Controllo e stabilita' di processi diffusivi nell'ambiente'', Polo delle Scienze e Tecnologie, Universita' degli Studi di Napoli Federico II  (2012).


\end{document}